\documentclass[a4paper]{article}
\pdfoutput=1
\usepackage{amsmath,amsfonts,amsthm,amssymb}
\usepackage{mathtools}
\usepackage{booktabs} 
\usepackage{multirow}
\usepackage[usenames,dvipsnames]{xcolor}
\usepackage{tikz-cd}
\usetikzlibrary{matrix,arrows,calc}
\usepackage{calc}
\usepackage[colorlinks=true,linkcolor=blue,citecolor=red,urlcolor=Violet,bookmarksdepth=subsection,final]{hyperref}
\usepackage[margin=3cm]{geometry}

\makeatletter
\renewcommand*\env@matrix[1][*\c@MaxMatrixCols c]{%
  \hskip -\arraycolsep
  \let\@ifnextchar\new@ifnextchar
  \array{#1}}
\makeatother

\tikzset{ampersand replacement=\&}

\allowdisplaybreaks[1] 

\newtheorem{thm}{Theorem}
\newtheorem{dfn}[thm]{Definition}
\newtheorem{prp}[thm]{Proposition}

\theoremstyle{remark}

  \newcommand{\del}{\partial}
  \newcommand{\oo}{\infty}

  \newcommand{\eps}{\varepsilon}
  \newcommand{\la}{\lambda}
  \newcommand{\ka}{\kappa}
  \newcommand{\tr}{\operatorname{tr}}
  \newcommand{\bnabla}{\bar{\nabla}}
  \newcommand{\sgn}{\operatorname{sgn}}
  \newcommand{\M}{\mathcal{M}}
\renewcommand{\S}{\mathcal{S}}

\title{IDEAL characterization of higher dimensional spherically symmetric black holes}

\author{Igor Khavkine\\[0.5ex]
	Institute of Mathematics, Czech Academy of Sciences,\\
	\v{Z}itn{\'a} 25, 115 67 Praha 1, Czech Republic\\[0.5ex]
	\texttt{khavkine@math.cas.cz}}

\date{October 6, 2018}

\begin{document}
\maketitle

\begin{abstract}
In general relativity, an IDEAL (Intrinsic, Deductive, Explicit,
ALgorithmic) characterization of a reference spacetime metric $g_0$
consists of a set of tensorial equations $T[g]=0$, constructed
covariantly out of the metric $g$, its Riemann curvature and their
derivatives, that are satisfied if and only if $g$ is locally isometric
to the reference spacetime metric $g_0$. We give the first IDEAL
characterization of generalized Schwarzschild-Tangherlini spacetimes,
which consist of $\Lambda$-vacuum extensions of higher dimensional
spherically symmetric black holes, as well as their versions where
spheres are replaced by flat or hyperbolic spaces. The standard
Schwarzschild black hole has been previously characterized in the work
of Ferrando and S\'aez, but using methods highly specific to
$4$ dimensions. Specialized to $4$ dimensions, our result provides an
independent, alternative characterization. We also give a proof of a
version of Birkhoff's theorem that is applicable also on neighborhoods
of horizon and horizon bifurcation points, which is necessary for our
arguments.
\end{abstract}

\section{Introduction}

In this work, we are interested in an intrinsic characterization of
higher dimensional generalizations of the Schwarzschild black hole.
These spherically symmetric, asymptotically flat vacuum spacetimes were
first studied by Tangherlini~\cite{tangherlini}. By a spacetime
$(\M,g)$, we mean a smooth manifold $\M$ with a Lorentzian metric $g$,
though a similar discussion can be carried out for any pseudo-Riemannian
geometry. While ``intrinsic'' generally does preclude direct reference
to the form of the spacetime metric in a special coordinate system, it
is a vague enough term to have multiple interpretations. To be specific,
we refer to an \emph{IDEAL}%
	\footnote{The acronym, explained in~\cite{fs-sphsym} (footnote, p.2),
	stands for Intrinsic, Deductive, Explicit and ALgorithmic.} %
or \emph{Rainich-type} characterization that has been used, for
instance, in the works~\cite{rainich, takeno, coll-ferrando,
fs-schw, fs-kerr, fs-sphsym, gpgl-schw, gpgl-kerr, krongos-torre, cdk}. It consists of
a list of tensorial equations ($T_k[g] = 0$, $a=1,2,\ldots,N$),
constructed covariantly out of the metric ($g$) and its derivatives
(concomitants of the Riemann tensor) that are satisfied if and only if
the given spacetime locally belongs to the desired class, possibly
narrow enough to be the isometry class of a single reference spacetime
geometry. This notion has a natural generalization ($T_k[g,\phi] = 0$)
to spacetimes equipped with scalar or tensor fields ($\phi$), with
equivalence still given by isometric diffeomorphisms that also transform
the additional scalars or tensors into each other, though we will not
make use of this generalization here. A nice historical survey of this
and other local characterization results can be found
in~\cite{mars-local}.

An IDEAL characterization requires neither the existence of any extra
geometric structures, nor the translation of the metric and of the
curvature into a frame formalism. Thus, it is an alternative to the
Cartan-Karlhede characterization~\cite[Ch.9]{stephani-sols}, which is
based on Cartan's moving frame formalism. Intrinsic characterizations, of
various types, have been of long standing and independent interest in
geometry and General Relativity. But, in addition, they can be helpful
in deciding when a metric, given for instance by some complicated
coordinate formulas, corresponds to one that is already known. In this
regard, an IDEAL characterization is especially helpful if one would
like to find an algorithmic solution to this recognition problem. In
numerical relativity, the near-satisfaction of the tensor equations
$T_k[g]\approx 0$ may signal the local proximity of a numerical
spacetime to a desired reference geometry. In addition, the approach to
zero $T_k[g] \to 0$ could be used to study either linear or nonlinear
stability of reference geometries, in an unambiguous and gauge
independent way.

The following additional application should be noted. By the
Stewart-Walker lemma~\cite[Lem.2.2]{stewart-walker}, the vanishing of a
tensor concomitant $T_k[g] = 0$ for a metric $g$ implies that its
linearization $\dot{T}_k[h]$ ($T_k[g+\eps h] = T_k[g] + \eps
\dot{T}_k[g] + O(\eps^2)$) is invariant under linearized
diffeomorphisms. Thus, any quantity of the form $\dot{T}_k[h]$ defines a
gauge-invariant observable in linearized gravity, when Einstein or
Einstein-matter equations are linearly perturbed about a background
solution $g$. The linearized Einstein tensor $\dot{G}[h]$ itself is of
course an example of a linear gauge-invariant, when linearized about a vacuum
solution, $G[g]=0$. In the presence of matter fields $\phi$, it is the
whole linearised Einstein equation $\dot{G}[h] - 8\pi\dot{T}[h,\psi]$
that is linearly gauge-invariant, when linearized about a joint solution of
$G[g]=8\pi T[g,\phi]$ and the matter equations of motion. And by
invariant, of course, we mean with respect to the action of linearized
diffeomorphisms on both the metric $g$ and the matter fields $\phi$.
A straightforward (though heuristic) argument shows that an
IDEAL characterization of a local isometry class provides a list
$\dot{T}_k[h]$, $k=1,\ldots,N$, of gauge-invariant observables that
should also complete: the joint kernel of $\dot{T}_k$ should coincide
with the tangent space to the isometry orbit (to make this argument
completely rigorous, it suffices to check that $T_k[g+h]$ do not
approach zero at $O(h^2)$ or higher order). That is, the joint kernel of
$\dot{T}_k[h]=0$ locally consists only of pure gauge modes
($h=\mathcal{L}_v g$ for some vector field $v$). The use of such local
observables (given by differential operators) can be advantageous both
in theoretical and practical investigations of classical and quantum
field theoretical models because they cleanly separate the local (or
ultraviolet) and global (or infrared) aspects of the theory. This may be
of interest in the problem of reconstructing the metric of a linear
gravitational wave from its complete set of gauge-invariant
observables~\cite{mobpm-kerr}, or in the problem of determining the
decay properties of linear gravitational waves in a gauge-independent
way~\cite{dotti-schw}.

In this work, we add the family of \emph{generalized
Schwarzschild-Tangherlini geometries} to the (unfortunately still small,
but slowly growing) literature concerning IDEAL characterizations of
isometry classes of individual reference geometries. That family
consists of all $\Lambda$-vacuum $2+m$-warped products, where the warped
$m$-dimensional factor is maximally symmetric. When the latter factor is
a round sphere, we recover the asymptotically flat or (anti-)de~Sitter
generalization of the Schwarzschild-Tangherlini black
holes~\cite{tangherlini, ki-master}. Replacing the sphere by flat
Euclidean space, we get the higher dimensional generalizations of
\emph{Taub's plane symmetric} spacetimes~\cite{taub}. Replacing the
sphere by hyperbolic space, we obtain so-called
\emph{pseudo-Schwarzschild wormhole} spacetimes~\cite{lobo-mimoso}.
Specializing to $2+m=4$ and $\Lambda=0$, we get the so-called family of
\emph{$A$-metrics} of Ehlers and Kundt~\cite{ehlers-kundt}.
Other IDEAL characterizations for geometries of interest in General
Relativity include ($4$-dimensional) Schwarzschild~\cite{fs-schw, gpgl-schw},
Reissner-Nordstr\"om~\cite{fs-typeD}, Kerr~\cite{fs-kerr, gpgl-kerr},
Lema\^itre-Tolman-Bondi~\cite{fs-sphsym}, Stephani
universes~\cite{fs-ricci} (see references for complete lists and
details), and most recently FLRW and inflationary spacetimes (in any
dimension)~\cite{cdk}. Of course, for completeness, we have to mention
the classic cases of constant curvature spaces (cf.~\eqref{eq:cc-riem}),
which are known to be fully characterized by the structure of the
Riemann tensor (by theorems of Riemann and Killing-Hopf~\cite{wolf-cc}).

For definiteness, let us state what we mean by a
\emph{local isometry} and \emph{local isometry class}.

\begin{dfn}[locally isometric] \label{def:loc-isom}
A pseudo-Riemannian geometry $(\M_1,g_1)$ is \emph{locally isometric at
$x_1\in \M_1$ to} a pseudo-Riemannian geometry $(\M_2,g_2)$ \emph{at
$x_2 \in \M_2$} if there exist open neighbourhoods $U_1 \ni x_1$, $U_2
\ni x_2$ and a diffeomorphism $\chi\colon U_1 \to U_2$ such that
$\chi(x_1) = x_2$ and $\chi^* g_2 = g_1$. If we can choose $U_1 = \M_1$
and $U_2 = \M_2$ then they are \emph{(globally) isometric}. If for every
$x_1\in \M_1$ there is $x_2 \in \M_2$ such that $(\M_1,g_1)$ at $x_1$ is
locally isometric to $(\M_2,g_2)$ at $x_2$, we simply say that
$(\M_1,g_1)$ is \emph{locally isometric to} $(\M_2,g_2)$ (note the
asymmetry in the definition). If $(\M_1,g_1)$ is locally isometric to
$(\M_2,g_2)$, as well as vice versa, we say that they are \emph{locally
isometric to each other} (which constitutes an equivalence relation).
All pseudo-Riemannian geometries that are locally isometric to a
reference $(\M,g)$ constitute its \emph{local isometry class}.
\end{dfn}

The synopsis of the paper is the following: In Section~\ref{sec:warped}
we define and exhibit the main geometric features of $2+m$-warped
product geometries. Proposition~\ref{prp:wp-char} gives a geometric
characterization of $2+m$-warped products in terms of a symmetric
projector whose covariant derivative satisfies a special constraint. In
Section~\ref{sec:gST} we introduce the family of \emph{generalized
Schwarzschild-Tangherlini} (gST) geometries, with special attention to the
structure of their Riemann curvature. Section~\ref{sec:birkhoff} states
and proves a version of Birkhoff's theorem, according to which a locally
maximally symmetric $2+m$-warped product that is also a $\Lambda$-vacuum
must locally coincide with one of the gST geometries. The main reason to
include a proof is to pay special attention to the applicability of this
result to points lying on a (Killing) horizon. Finally,
Theorem~\ref{thm:ideal-gST} in Section~\ref{sec:ideal} puts all the
pieces together to give an IDEAL characterization of the local isometry
classes of the gST geometries. Due to a quirk of the structure of the
gST Riemann curvature in $n=4$ dimensions, the final result looks
slightly different in $n=4$ and $n\ge 5$ dimensions. This difference is
accounted for by Theorem~\ref{thm:ideal-gST-4d}. In the case of
spherical symmetry in $n=4$ dimensions with $\Lambda=0$, our results
provide an independent alternative characterization of the standard
Schwarzschild spacetime, which was first characterized
in~\cite{fs-schw}. All other instances of the results from
Section~\ref{sec:ideal} are new. Finally, in Section~\ref{sec:discuss},
we conclude with a discussion of our results and of directions for
future work.

Throughout the paper we follow the conventions of~\cite{wald-gr}:
$({-}{+} \cdots {+})$ for Lorentzian signature, and $2\nabla_{[\mu}
\nabla_{\nu]} \omega_\la = R_{\mu\nu\la}{}^\ka \omega_\ka$ for
curvature. Unless otherwise specified, all functions will be considered
$C^\oo$ smooth.

\section{$2+m$-warped products} \label{sec:warped}

Below, we consider $2+m$-warped product geometries. That is,
pseudo-Riemannian geometries on an $n$-dimensional manifold, which can
be represented as a warped product of a $2$-dimensional and an
$m$-dimensional geometry. Our main family of examples consists of the
\emph{generalized Schwarzschild-Tangherlini} spacetimes
(Section~\ref{sec:gST}), which includes the spherically symmetric black
holes in four and higher dimensions. We will discuss the structure of
Riemann curvature tensor of the warped product and the consequences for
the $2$-dimensional base factor when the product satisfies the Einstein
equation (\emph{Birkhoff's theorem}, Section~\ref{sec:birkhoff}).

\begin{dfn}[warped product]
A pseudo-Riemannian geometry $(\bar{\M},\bar{g}) \cong (\M,g) \times_r
(\S,\Omega)$ is a \emph{warped product} with \emph{warping function $r$}
when $\bar{\M} \cong \M \times \S$ and the metric can be written as
\begin{equation} \label{eq:wp}
	\bar{g} = g + r^2 \Omega,
\end{equation}
where the metric tensors $g$ and $\Omega$
are lifted to the product space by pulling back along the projections
$\bar{\M} \to \M$ and $\bar{\M} \to \S$, while $r$ is the pullback of a
nowhere vanishing function on $\M$. We call $(\S,\Omega)$ the
\emph{warped factor} and $(\M,g)$ the \emph{base factor}.
\end{dfn}

Let us now introduce some notational conventions that will simplify
subsequent discussions. Denote by $\bnabla_\mu$, $\nabla_a$ and $D_A$
the canonical Levi-Civita connections on $(\bar{\M},\bar{g}) \cong
(\M,g) \times_r (\S,\Omega)$, $(\M,g)$ and $(\S,\Omega)$, respectively.
We will use Greek indices $(\alpha\beta\cdots)$ for tensors on
$\bar{\M}$, lower case Latin indices $(ab\cdots)$ for tensors on $\M$,
and upper case Latin indices $(AB\cdots)$ on $\S$. Using the product
structure $\bar{\M} \cong \M \times \S$, any tensor or differential
operator on $\M$ or $\S$ can be canonically transferred to $\bar{\M}$.
We will do so for the usual Riemann and Ricci curvature tensors and the
derivative of the warping function:
\begin{equation}
\begin{gathered}
	R_{abcd}[g], R_{ab}[g], R[g] \to R_{\mu\nu\la\ka}, R_{\mu\nu}, R ,
		\quad
	R_{ABCD}[\Omega], R_{AB}[\Omega], R[\Omega] \to S_{\mu\nu\la\ka}, S_{\mu\nu}, S ,
		\\
	r_a = \nabla_a r \to r_\mu = \nabla_\mu r .
\end{gathered}
\end{equation}
We will also use the obvious notation $\bar{R}_{\mu\nu\la\ka} =
R_{\mu\nu\la\ka}[\bar{g}]$, $\bar{R}_{\mu\nu} = R_{\mu\nu}[\bar{g}]$,
$\bar{R} = R[\bar{g}]$. We will use the following self-explanatory
convention when raising and lowering indices on tensors transferred to
$\bar{\M}$ from one of the factors: $\bar{g}^{\mu\nu} g_{\nu\la} =
g^{\mu}{}_{\la}$, but $\bar{g}^{\mu\nu} r^2 \Omega_{\nu\la}
\bar{g}^{\la\ka} = r^{-2} \Omega^{\mu\ka}$.

The covariant derivative on the warped product geometry acts as
\begin{equation}
	\bnabla_\mu X_\nu = \nabla_\mu X_\nu + D_\mu X_\nu
		- 2 X^\la (r^2 \Omega)_{\la(\nu} \bnabla_{\mu)} \log r
		+ X^\la (\bnabla_\la \log r)(r^2 \Omega)_{\mu\nu} ,
\end{equation}
where the action of $\bnabla$ on scalars is just through the exterior
derivative. In particular, we get
\begin{equation} \label{eq:wp-gradg}
	\bnabla_\mu g_{\nu\la}
	= 2 \frac{r^2 \Omega_{\mu(\nu} r_{\la)}}{r} ,
	\quad \text{and} \quad
	\bnabla_\mu (r^2 \Omega)_{\nu\la}
	= -2 \frac{r^2 \Omega_{\mu(\nu} r_{\la)}}{r} .
\end{equation}
The curvature tensors are given by
\begin{align}
\label{eq:wp-riem}
	\bar{R}_{\mu\nu\la\ka}
	&= r^2 S_{\mu\nu\la\ka}
		+ R_{\mu\nu\la\ka}
		- \left(\frac{\nabla\nabla r}{r} \odot r^2 \Omega\right)_{\mu\nu\la\ka}
		- \frac{r_\sigma r^\sigma}{2r^2} (r^2\Omega \odot r^2\Omega)_{\mu\nu\la\ka} , \\
\label{eq:wp-ricc}
	\bar{R}_{\mu\nu}
	&= S_{\mu\nu} + R_{\mu\nu}
		- m\frac{\nabla_\mu\nabla_\nu r}{r} - \frac{\square r}{r} r^2\Omega_{\mu\nu}
		- (m-1) \frac{r_\sigma r^\sigma}{r^2} r^2 \Omega_{\mu\nu} , \\
	\bar{R}
	&= \frac{S}{r^2} + R - 2m \frac{\square r}{r} - m(m-1) \frac{r_\sigma r^\sigma}{r^2} ,
\end{align}
where for convenience we have introduced the \emph{Kulkarni-Nomizu
product} of symmetric tensors:
\begin{equation}\label{kn-product}
	(A \odot B)_{\mu\nu\la\ka}
	= A_{\mu\la} B_{\nu\ka} - A_{\nu\la} B_{\mu\ka}
	- A_{\mu\ka} B_{\nu\la} + A_{\nu\ka} B_{\mu\la} .
\end{equation}
Note that we could rewrite the $\nabla$-derivatives using
$\bnabla$-derivatives in the above formulas with the help of the
identity
\begin{equation} \label{eq:ddr}
	\frac{\nabla_\mu \nabla_\nu r}{r}
	= \frac{\bnabla_\mu \bnabla_\nu r}{r}
		- \frac{r_\la r^\la}{r^2} (r^2 \Omega)_{\mu\nu}
	= \bnabla_\mu \frac{\bnabla_\nu r}{r}
		+ \frac{r_\mu r_\nu}{r^2}
		- \frac{r_\la r^\la}{r^2} (r^2\Omega)_{\mu\nu} .
\end{equation}

These formulas can be extracted from~\cite[Prp.7.42]{oneill}. But the
quickest way check them is to notice that $\bar{g}_{\mu\nu} = r^2
(r^{-2} g_{\mu\nu} + \Omega_{\mu\nu})$, in which a pair of nested
conformal transformation relates $\bar{g}$ and a product metric, and use
the standard formulas for the conformal transformations of covariant
derivatives and curvatures~\cite[App.D]{wald-gr}.

\begin{prp}[\cite{fs-2+2}, {\cite[Thm.16(2)]{gpgl-biconf}}] \label{prp:wp-char}
A pseudo-Riemannian geometry $(\bar{\mathcal{M}}, \bar{g})$ can be
locally put into the form of a $2+m$-warped product
\begin{equation}
	\bar{g}_{\mu\nu} = g_{\mu\nu} + r^2 \Omega_{\mu\nu}
\end{equation}
iff there exist a $1$-form $\ell_\mu$ and a symmetric tensor
$\bar{\Omega}_{\mu\nu} = \bar{\Omega}_{(\mu\nu)}$ that together satisfy
the following conditions
\begin{equation} \label{eq:wp-char}
\begin{gathered}
	\bnabla_{[\mu} \ell_{\nu]} = 0 ,
		\quad
	\ell^\mu \bar{\Omega}_{\mu\nu} = 0 ,
		\\
	\bar{\Omega}_{\mu}{}^\nu \bar{\Omega}_{\nu\la} = \bar{\Omega}_{\mu\la} ,
		\quad
	\bar{\Omega}_\mu{}^\mu = m ,
		\quad
	\bnabla_\mu \bar{\Omega}_{\nu\la} = -2 \bar{\Omega}_{\mu(\nu} \ell_{\la)} .
\end{gathered}
\end{equation}
Then we can choose $g_{\mu\nu} = \bar{g}_{\mu\nu} -
\bar{\Omega}_{\mu\nu}$ and $\Omega_{\mu\nu} = r^{-2} \bar{\Omega}$, with
$r$ satisfying $\bnabla_\mu \log |r| = \ell_\mu$.
\end{prp}

\begin{proof}
In one direction, starting with the definitions of
$\bar{\Omega}_{\mu\nu}$ and $\ell_\mu$ in terms of $r$ and
$\Omega_{\mu\nu}$, verifying the above identities is a matter of direct
calculation (cf. Equation~\eqref{eq:wp-gradg}).

In the other direction, the key observation is that a warped product
metric is conformal to a direct product metric, namely $\bar{g}_{\mu\nu}
= r^2 \hat{g}_{\mu\nu}$, with $\hat{g}_{\mu\nu} = r^{-2} g_{\mu\nu} +
\Omega_{\mu\nu}$, where the warping function $r^2$ plays the role of the
conformal factor. Given our first condition on $\ell_\mu$, and since we
are working locally, we can always find a smooth function $r$ satisfying
$\bnabla_\mu \log |r| = \bnabla_\mu r/r = \ell_\mu$. Again locally, we
can choose $r$ to be nowhere vanishing. The choice is unique, up to a
multiplicative constant.

Define $\hat{g}_{\mu\nu} = r^{-2} \bar{g}_{\mu\nu}$, $\hat{g}^{\mu\nu}$
its inverse, $\Omega_{\mu\nu} = r^{-2} \bar{\Omega}_{\mu\nu}$ and let
$\hat{\nabla}$ be the $\hat{g}$-compatible Levi-Civita connection. A
straightforward calculation shows that our conditions on $\bar{\Omega}$
translate to
\begin{equation}
	\Omega_{\mu\nu} \hat{g}^{\nu\la} \Omega_{\la\ka} = \Omega_{\mu\ka} ,
	\quad
	\hat{g}^{\mu\nu} \Omega_{\mu\nu} = m ,
	\quad \text{and} \quad
	\hat{\nabla}_\mu \Omega_{\nu\la} = 0 .
\end{equation}
It is well known~\cite{fs-2+2} that the existence of such a rank-$m$
$\hat{\nabla}$-covariantly constant symmetric projector
$\Omega_{\mu\nu}$ implies that $\hat{g}_{\mu\nu}$ can be locally put
into $2+m$-product form $\hat{g}_{\mu\nu} = (\hat{g}_{\mu\nu} -
\Omega_{\mu\nu}) + \Omega_{\mu\nu}$. Our second condition on $\ell_\mu$
implies that $r$ does not depend on the $\Omega$-factor. Thus, without
loss of generality, we can define $g_{\mu\nu} = r^2 (\hat{g}_{\mu\nu} -
\Omega_{\mu\nu})$ and write the product form as $\hat{g}_{\mu\nu} =
r^{-2} g_{\mu\nu} + \Omega_{\mu\nu}$.

Undoing the conformal transformation, we end up with the desired local
$2+m$-warped product form $\bar{g}_{\mu\nu} = g_{\mu\nu} + r^2
\Omega_{\mu\nu}$.
\end{proof}

\subsection{Generalized Schwarzschild-Tangherlini geometries} \label{sec:gST}

Consider an integer $n\ge 4$ and a triple of real numbers
$(\alpha,M,\Lambda)$, where $M\ne 0$. The $2$-dimensional metric
\begin{equation} \label{eq:schw-f}
	g_{ab} = -f dt_a dt_b + \frac{1}{f} dr_a dr_b ,
	\quad
	f(r) = \alpha - \frac{2M}{r^{n-3}} - \frac{2\Lambda}{(n-1)(n-2)} r^2 ,
\end{equation}
is well-defined and Lorentzian in the interiors of the $r$-intervals
separated by $r=0$ and the roots of $f(r)=0$. It is well-known that each
one of these regions has a unique maximal analytic, connected and
simply-connected extension~\cite{lake, schleich-witt}. Each region with
$r>0$ generates the same extension (topologically $\mathbb{R}^2$), and
similarly for each region with $r<0$. When $n$ is even, the extensions
with $r>0$ and $r<0$ are distinct. However, the $r<0$ extension is
isometric to the $r>0$ extension with $M$ replaced by $-M$, by sending
$r\mapsto -r$. When $n$ is odd, the extensions with $r>0$ and $r<0$ are
isometric, again by sending $r\mapsto -r$, but the geometry obtained by
replacing $M$ by $-M$ is different. Thus, for book keeping convenience,
let us denote by $(\M,g)_{n,\alpha,M,\Lambda}$ the disjoint union of the
$r>0$ and $r<0$ extensions with the same $M$ parameter when $n$ is even,
and the disjoint union of the $r>0$ extension with parameter $M$ and the
$r<0$ extension with parameter $-M$ when $n$ is odd. In either case, $\M
\cong \mathbb{R}^2 \sqcup \mathbb{R}^2$. Naturally, by our construction,
each of these maximally extended geometries is accompanied by the
distinguished scalar function $r$, taking on all non-zero real values,
which was analytically extended along with the metric.

The precise way in which the $(t,r)$ charts are glued together along
\emph{horizons} and \emph{horizon bifurcation} points to form the
analytic extension can be glimpsed from Proposition~\ref{prp:birkhoff},
where (a) corresponds to a \emph{generic} points covered by an $(t,r)$
chart, (b) corresponds to a \emph{horizon} points, and (c) corresponds
to a \emph{horizon bifurcation} point of the extension. The gluing is
done with the help of the \emph{tortoise coordinate} $r_*$
from~\eqref{eq:r-tortoise}. Penrose conformal diagrams for the
extensions can be found in~\cite{lake, schleich-witt}.

Recall that $(\S,\Omega)$ is of \emph{constant
curvature}~\cite{wolf-cc}, with sectional curvature $\alpha$, if its
Riemann curvature tensor is
\begin{equation} \label{eq:cc-riem}
	R_{ABCD}[\Omega]
	= \frac{\alpha}{2} (\Omega \odot \Omega)_{ABCD} .
\end{equation}
When an $m$-dimensional Riemannian geometry $(\S,\Omega)$ is simply
connected, geodesically complete and of constant curvature it can only
be one of the following~\cite[Sec.2.4]{wolf-cc}: Euclidean $m$-space,
round $m$-sphere, hyperbolic $m$-space. These are called \emph{maximally
symmetric} spaces. Let us denote the corresponding maximally symmetric
space with sectional curvature $\alpha$ by $(\S,\Omega)_{m,\alpha}$.

\begin{dfn}[generalized Schwarzschild-Tangherlini spacetime] \label{def:gST}
Fix a dimension $m\ge 2$ and a triple of real numbers
$(\alpha,M,\Lambda)$, with $M\ne 0$. Set $n=2+m$ and denote the warped
product $(\bar{\M}, \bar{g})_{\alpha,M,\Lambda} \cong
(\M,g)_{(n,\alpha,M,\Lambda)} \times_r (\S,\Omega)_{m,\alpha}$, where the
base factor and the warping function $r$ are defined in the discussion
following~\eqref{eq:schw-f}, and the warped factor is the
$m$-dimensional maximally symmetric space of sectional curvature
$\alpha$. We call $(\bar{\M},\bar{g})_{\alpha,M,\Lambda}$ a
$n$-dimensional \emph{generalized Schwarzschild-Tangherlini (gST)
spacetime}.
\end{dfn}

If we had included the $M=0$ cases, then each such geometry would
correspond to a particular representation of a subset of a maximally
symmetric geometry (de~Sitter or anti-de~Sitter spacetime), as we will
see shortly (Equation~\eqref{eq:gST-riem}). Since this case has already
been extensively studied (e.g.,\ see our previous
works~\cite{kh-calabi,cdk}), we exclude it from consideration.

For tensors with two or four indices, we define contractions
\begin{equation} \label{eq:dot-con}
	(A\cdot B)_{\mu\nu}
	= A_{\mu}{}^\lambda B_{\lambda\nu} ,
	\quad \text{and} \quad
	(R \cdot S)_{\mu\nu\lambda\kappa}
	= R_{\mu\nu}{}^{\sigma\tau} S_{\sigma\tau\lambda\kappa} .
\end{equation}
Recalling also the definition of the Kulkarni-Nomizu
product~\eqref{kn-product}, when $A$, $B$, $C$ and $D$ are symmetric, we
have the useful identities
\begin{gather}
	[(A\odot B) \cdot (C\odot D)]_{\mu\nu\lambda\kappa}
	= 2[(A\cdot C)\odot(B\cdot D) + (A\cdot D)\odot(B\cdot
	C)]_{\mu\nu\lambda\kappa} , \\
	(A\odot B)_{\mu\nu}{}^\nu{}_\kappa
	= [A\cdot B - (\tr A) B - A (\tr B) + B\cdot A]_{\mu\kappa} .
\end{gather}

Now, we compute the curvature of the gST geometries that we have defined
above. Let us start with the $2$-dimensional $(\M,g)$ factor. We
basically follow the presentation from~\cite{martel-poisson}. Working in
the $(t,r)$ chart, clearly $t^a = (\del_t)^a$ is a timelike Killing
vector. For convenience, we also introduce the notation $t_a = g_{ab}
t^b = -f dt_a$ and $r_a = dr_a$. They are related as $t^a = -\eps^{ab}
r_b$, where $\eps_{ab} = (dt \wedge dr)_{ab}$. Then, of course, $r_a r^a
= f$ and $t_a t^a = -1/f$. The action of the covariant derivative is
summarized by
\begin{equation} \label{eq:schw-cov-frame}
	\nabla_a t_b = \frac{f_1}{2r} \eps_{ab}
	\quad \text{and} \quad
	\nabla_a r_b = \frac{f_1}{2r} g_{ab} .
\end{equation}

For the record, let us write out in full the following identities for
$(\M,g)$:
\begin{align}
	\frac{r_a r^a}{r^2}
	&= \frac{f}{r^2}
	= \frac{\alpha}{r^2} - 2\frac{M}{r^{n-1}} - \frac{2\Lambda}{(n-1)(n-2)} , \\
	\frac{\nabla_a \nabla_b r}{r}
	&= \left((n-3)\frac{M}{r^{n-1}}
		- \frac{2\Lambda}{(n-1)(n-2)}\right) g_{ab} , \\
	R_{abcd} &= \frac{R}{4} (g\odot g)_{abcd} , \quad
	R_{ab} = \frac{R}{2} g_{ab} , \\
	R &= \frac{4\Lambda}{(n-1)(n-2)} + 2(n-2)(n-3) \frac{M}{r^{n-1}} .
\end{align}
They can be directly plugged into~\eqref{eq:wp-riem}, the formula for
the Riemann tensor of a $2+m$-warped product, to get the explicit
expression for the Riemann tensor $\bar{R}_{\mu\nu\la\ka}$ of a gST
geometry $(\bar{\M},\bar{g})$.

\begin{multline} \label{eq:gST-riem}
	\bar{R}_{\mu\nu\la\ka} =
		\frac{\Lambda}{(n-1)(n-2)} (\bar{g}\odot \bar{g})_{\mu\nu\la\ka} \\
	+ \frac{M}{r^{n-1}} \left[
		\frac{(n-2)(n-3)}{2} (g\odot g)_{\mu\nu\la\ka}
	+ (r^2\Omega \odot r^2\Omega)_{\mu\nu\la\ka}
	- (n-3) (g\odot r^2\Omega)_{\mu\nu\la\ka} \right] .
\end{multline}

Next, let us define several tensors and scalars built out of the Riemann
tensor and its derivatives:
\begin{align}
\label{eq:gST-T}
	\bar{T}_{\mu\nu\la\ka}[\bar{g}]
	&:= \bar{R}_{\mu\nu\la\ka}
		- \frac{\Lambda}{(n-1)(n-2)} (\bar{g}\odot \bar{g})_{\mu\nu\la\ka} , \\
\label{eq:gST-rho}
	\rho[\bar{g}]
	&:= \left[\frac{(\bar{T}\cdot \bar{T} \cdot \bar{T})_{\mu\nu}{}^{\mu\nu}}
			{8(n-1)(n-2)(n-3) [(n-2)(n-3)(n-4)+2]}\right]^{\frac{1}{3}} , \\
\label{eq:gST-ell}
	\ell_\mu[\bar{g}]
	&:= -\frac{1}{(n-1)} \frac{\bnabla_\mu \rho}{\rho} , \\
\label{eq:gST-A}
	A[\bar{g}]
	&:= \ell_\mu \ell^\mu + 2\rho + \frac{2\Lambda}{(n-1)(n-2)} .
\end{align}
For future reference, we also compute some algebraic combinations among
these tensors (see~\cite[Sec.3.3]{kh-compat} for more intermediate steps
of the calculations):
\begin{align}
	\bar{T}
	&= \frac{M}{r^{n-1}} \left[
		\frac{(n-2)(n-3)}{2} (g\odot g)
	+ (r^2\Omega \odot r^2\Omega)
	- (n-3) (g\odot r^2\Omega) \right] , \\
	\bar{T}\cdot \bar{T}
	&= \left(\frac{M}{r^{n-1}}\right)^2 \left[
		(n-2)^2(n-3)^2 (g\odot g)
		+ 4 (r^2\Omega\odot r^2\Omega)
		+ 2 (n-3)^2 (g\odot r^2\Omega) \right] \\
	\bar{T}\cdot \bar{T} \cdot \bar{T}
	&= \left(\frac{M}{r^{n-1}}\right)^3 \left[
		2(n-2)^3(n-3)^3 (g\odot g)
		+ 16 (r^2\Omega\odot r^2\Omega)
		- 4 (n-3)^3 (g\odot r^2\Omega) \right] , \\
	(\bar{T}\cdot \bar{T})_{\mu\nu}{}^{\nu}{}_\ka
	&= -\left(\frac{M}{r^{n-1}}\right)^2 \left[
		2(n-1)(n-2)(n-3)^2 g_{\mu\ka}
		+ 4(n-1)(n-3) r^2\Omega_{\mu\ka} \right] , \\
	(\bar{T}\cdot \bar{T})_{\mu\nu}{}^{\mu\nu}
	&= 4(n-1)(n-2)^2(n-3) \left(\frac{M}{r^{n-1}}\right)^2 , \\
	(\bar{T}\cdot \bar{T} \cdot \bar{T})_{\mu\nu}{}^{\mu\nu}
	&= 8(n-1)(n-2)(n-3) [(n-2)(n-3)(n-4)+2] \left(\frac{M}{r^{n-1}}\right)^3 ,
	\\
	&\qquad
	\rho = \frac{M}{r^{n-1}} , \quad
	\ell_\mu = \frac{r_\mu}{r} , \quad
	A = \frac{\alpha}{r^2} .
\end{align}
Given the last row of identities, it is clear that
\begin{equation} \label{eq:gST-Arho}
	\sgn A = \sgn \alpha, \quad \text{and} \quad
	A^{n-1} \rho^{-2} = \alpha^{n-1} M^{-2} .
\end{equation}

Next, we find a way to express the projector $r^2\Omega_{\mu\nu}$ onto
the warped factor in terms of the curvature. Here we find a slight
dimension dependence (as already noted in~\cite[Sec.3.3]{kh-compat}). In
dimension $n\ge 5$, one can find a formula that involves only products
and contractions $\bar{g}$ of $\bar{T}$:
\begin{align} \label{eq:gST-Omega}
	(r^2 \Omega)_{\la\ka}
	&=
	\frac{2(n-2)^2}{(n-1)(n-4)} 
	\frac{(\bar{T}\cdot \bar{T})_{\la\nu}{}^{\nu}{}_\ka}{(\bar{T}\cdot\bar{T})_{\mu\nu}{}^{\mu\nu}}
	+ \frac{(n-2)(n-3)}{(n-1)(n-4)} \bar{g}_{\la\ka} .
\end{align}
Obviously, the above formula has poles and hence fails when $n=4$. On
the other hand, the following slightly more complex formula works both
in $n=4$ as well as higher dimensions:
\begin{equation} \label{eq:gST-Omega-4d}
	(r^2 \Omega)_{\la\ka}
	=
	-\frac{1}{(n-1)(n-3) \, \rho \, \ell^2} \left(\bar{T}_{\mu\la\nu\ka}
		- \frac{(n-2)(n-3)}{2} \rho \, (\bar{g}\odot \bar{g})_{\mu\la\nu\ka} \right)
		\ell^\la \ell^\ka .
\end{equation}
The complexity of the second formula is due to the presence of the
$\ell_\mu$ vector, which is itself defined as the gradient of a scalar
of a gradient constructed from $\bar{T}$. Thus
formula~\eqref{eq:gST-Omega} may be preferable in $n\ge 5$
to~\eqref{eq:gST-Omega-4d}, even if the latter also works in higher
dimensions.

The following result simply identifies the invariant parameters that can
be used to exhaustively label the distinct isometry classes of gST
reference geometries as we have defined them earlier. Note that below we
adopt the convention that the sign function satisfies $\sgn 0 = 0$.

\begin{prp} \label{prp:gST-isom}
A gST geometry $(\bar{\M}, \bar{g})_{\alpha,M,\Lambda}$ is locally
isometric at $x\in \bar{\M}$ to another gST geometry $(\bar{\M}',
\bar{g}')_{\alpha',M',\Lambda'}$ at $x' \in \bar{\M}'$ iff
\begin{equation}
	((\sgn \alpha) \left|\alpha\right|^{n-1} M^{-2}, \Lambda)
	= ((\sgn \alpha') \left|\alpha'\right|^{n-1} M'^{-2}, \Lambda') ,
		\quad
	\rho[\bar{g}](x) = \rho[\bar{g}'](x')
\end{equation}
and one of the following holds
\begin{itemize}
\item[(a)] $x$ is a generic point,
	$\ell^2[\bar{g}](x) \ne 0 \ne \ell^2[\bar{g}'](x')$, or
\item[(b)] $x$ is a horizon point,
	$\ell_\mu[\bar{g}](x) \ne 0 \ne \ell_\mu[\bar{g}'](x')$
	and $\ell^2[\bar{g}](x) = 0 = \ell^2[\bar{g}'](x')$, or
\item[(c)] $x$ is a horizon bifurcation point,
	$\ell_\mu[\bar{g}](x) = 0 = \ell_\mu[\bar{g}'](x')$.
\end{itemize}
Hence, $(\bar{\M}, \bar{g})_{\alpha,M,\Lambda}$ and $(\bar{\M}',
\bar{g}')_{\alpha',M',\Lambda'}$ are isometric iff
\begin{equation}
	((\sgn \alpha) \left|\alpha\right|^{n-1} M^{-2}, \Lambda)
	= ((\sgn \alpha') \left|\alpha'\right|^{n-1} M'^{-2}, \Lambda'^2) .
\end{equation}
\end{prp}

In other words, the pair or real numbers
\begin{equation}
	\left((\sgn \alpha) |\alpha|^{n-1} M^{-2}, \Lambda\right) \in \mathbb{R}^2
\end{equation}
uniquely and exhaustively identifies all isometry classes among the gST
geometries (Definition~\ref{def:gST}). Moreover, non-isometric gST
geometries are not even locally isometric.

\begin{proof}
Many of the arguments below are based on the fact that the existence of
a local isometry linking the points $x$ and $x'$ forces the pairwise
equality of all curvature scalars respectively evaluated at these
points. A slight generalization of this idea to covariant identities
involving curvature tensors immediately establishes our claims (a), (b)
and (c). Also, recall that, from our definition of a reference gST
geometry, the coordinate transformation $r\mapsto -r$ always corresponds
to the parameter flip $M \mapsto -M$, independent of the parity of the
dimension $n$. Finally, we will assume that the points $x\in \bar{\M}$
and $x' = \bar{\M'}$ belong to the regions where we can introduce the
$(t,r)$ and $(t',r')$ coordinates, as in~\eqref{eq:schw-f}, on the base
factors. Then, simple coordinate transformations on $(t,r)$ extend to
globally defined diffeomorphisms of $\bar{\M}$ or $\bar{\M}'$ by
analyticity. When such coordinates are ill-defined on neighborhoods of
$x$ or $x'$, the same logic applies, but where we need to directly apply
the diffeomorphisms defined by analytic extension.

First, note that $\Lambda = \Lambda'$ is necessary for local isometry.
Relying on~\eqref{eq:gST-riem}, we can obtain this constant from the
Ricci scalar of the geometry, $\bar{R}[\bar{g}] =
\frac{2n\Lambda}{(n-2)}$ and $\bar{R}[\bar{g}'] =
\frac{2n\Lambda'}{(n-2)}$. Let us assume the equality $\Lambda =
\Lambda'$ from now on.

Next, relying on~\eqref{eq:gST-Arho}, note that
\begin{gather}
	\sgn A[\bar{g}] = \sgn \alpha, \quad
	A[\bar{g}]^{n-1} \rho[\bar{g}]^{-2} = \alpha^{n-1} M^{-2}
	\\ \text{and} \quad
	\sgn A[\bar{g}'] = \sgn \alpha', \quad
	A[\bar{g}']^{n-1} \rho[\bar{g}']^{-2} = \alpha'^{n-1} M'^{-2} .
\end{gather}
Hence, since knowledge of $(\sgn \alpha) |\alpha|^{n-1} M^{-2}$ is
equivalent to the knowledge of both $\sgn \alpha$ and $\alpha^{n-1}
M^{-2}$, the equality $(\sgn \alpha) |\alpha|^{n-1} M^{-2} = (\sgn
\alpha') |\alpha'|^{n-1} M^{-2}$ is also necessary for local isometry.
Let us assume this equality from now on. It remains only to check that
both equalities guarantee the existence of an isometry.

When $\sgn \alpha = 0 = \sgn \alpha'$, apply the coordinate
transformations $(1/t,r) \mapsto |M|^{\frac{1}{n-1}} (1/t,r)$ and
$(1/t',r') \mapsto |M'|^{\frac{1}{n-1}} (1/t',r')$, together with a
possible $r\mapsto -r$ and/or $r'\to -r'$ flip, depending on the signs
of $M$ and $M'$, to bring the base factors to isometric form with
$M=1=M'$. To keep $r$ and $r'$ as the warping functions, these
transformations must be accompanied by the conformal rescalings $\Omega
\mapsto |M|^{\frac{2}{n-1}} \Omega$ and $\Omega' \mapsto
|M|^{\frac{2}{n-1}} \Omega'$. However, since $\alpha = 0 = \alpha'$ and
the two warped factors are flat, these rescalings do not affect their
isometry class. Hence, the two gST geometries must be isometric since
they have the same warped product structure.

Now, assume that $\alpha \ne 0 \ne \alpha'$, while necessarily $\sgn
\alpha = (-1)^k = \sgn \alpha'$. Then the coordinate redefinitions
$(1/t,r) \mapsto |\alpha|^{\frac{1}{2}} (1/t,r)$ and $(1/t',r') \mapsto
|\alpha|^{\frac{1}{2}} (1/t',r')$ bring them both to $\alpha = (-1)^k =
\alpha'$. Let us assume this equality from now on.

The only possible remaining difference between the parameters is that
$\sgn M \ne \sgn M'$, while $|M| = |M'|$. But then, applying $r\mapsto
-r$ or $r' \mapsto -r'$ brings about the equality $M = M'$ and hence the
desired isometry.

Clearly, the above arguments apply both to local isometries as well as
to global isometries. This concludes the proof.
\end{proof}

\subsection{Birkhoff's theorem} \label{sec:birkhoff}

It is well-known that being a $2+m$-warped product solution of
cosmological Einstein's equations is highly restrictive. In particular,
the geometry of the base factor is restricted to locally coincide with
one of the base factors of a gST geometry, whether the warped factor is
spherically symmetric or not. This rigidity result (though usually
stated with the spherical symmetry assumption) is known as Birkhoff's
theorem~\cite{lake, schmidt, schleich-witt, an-wong}. Below, we state
and prove a version that is convenient for our purposes. The main reason
to include a proof is to make sure that we can cover the corner cases
(when $\nabla r$ becomes null or even vanishes) that are often skipped
in the literature.

Recall that a metric $\bar{g}_{\mu\nu}$ is called a
\emph{$\Lambda$-vacuum} when it satisfies the source-free Einstein
equations with cosmological constant $\Lambda$:
\begin{equation} \label{eq:ee}
	\bar{R}_{\mu\nu}[\bar{g}] - \frac{1}{2} \bar{R}[\bar{g}] \bar{g}_{\mu\nu}
		+ \Lambda \bar{g}_{\mu\nu} = 0
	\iff
	\bar{R}_{\mu\nu}[\bar{g}] - \frac{2\Lambda}{(n-2)} \bar{g}_{\mu\nu} = 0 .
\end{equation}

\begin{prp}[\cite{lake, schmidt, schleich-witt, an-wong}] \label{prp:birkhoff}
Consider a pseudo-Riemannian geometry $(\bar{\M}, \bar{g})$ of dimension
$n=2+m$ that locally, say at $\bar{x}\in \bar{\M}$, has the form $(\M,g)
\times_r (\S,\Omega)$ of a $2+m$-warped product, with
\begin{equation}
	\bar{g}_{\mu\nu} = g_{\mu\nu} + r^2 \Omega_{\mu\nu} ,
\end{equation}
where $(\M,g)$ is Lorentzian and $r$ is not locally constant at $x\in
\M$, the projection of $\bar{x}$ to $\M$. When $\bar{g}_{\mu\nu}$ is a
$\Lambda$-vacuum that is not locally of constant curvature at $\bar{x}$,
the metric of the base factor can be locally put into one of the
following forms at $x$:
\begin{itemize}
\item[(a)]
	when $\nabla_a r \ne 0$, $(\nabla r)^2 \ne 0$ at $x$, in local $(t,r)$
	coordinates,
	\begin{equation}
		g_{ab} = -f \, dt_a dt_b + \frac{1}{f} \, dr_a dr_b ,
	\end{equation}
\item[(b)]
	when $\nabla_a r \ne 0$, $(\nabla r)^2 = 0$, $r=r_H$ at $x$, in local
	$(v,r)$ coordinates,
	\begin{equation}
		g_{ab} = -f \, dv_a dv_b + 2 \, dv_{(a} dr_{b)} ,
	\end{equation}
\item[(c)]
	when $\nabla_a r = 0$, $r = r_H$ at $x$, in local $(U,V)$ coordinates,
	\begin{gather}
	\notag
		r = r(UV) = r_H + r_H UV + O^3(U,V) , \quad \text{with} \quad
		z r'(z) = \frac{f(r)}{f'(r_H)} ,
		\quad \text{and} \\
	\label{eq:birkhoff-c}
		g_{ab} = \frac{-4f e^{-h}}{f'(r_H)^2 \, (1-r/r_H)} \, dU_{(a} dV_{b)} ,
		\quad \text{with} \quad
		h(r) = \int_{r_H}^r ds \left(\frac{f'(r_H)}{f(s)} - \frac{1}{s-r_H}\right) ,
	\end{gather}
\end{itemize}
where in each case
\begin{equation} \label{eq:f-birkhoff}
	f(r) = \alpha - \frac{2M}{r^{n-3}} - \frac{2\Lambda}{(n-1)(n-2)} r^2 ,
\end{equation}
for some constants $\alpha$ and $M \ne 0$. In cases (b) and (c), $r=r_H$
is a root of $f(r)=0$; in case (c) the root is always simple.

Thus, $(\M,g)$ is locally isometric at $x$ to either (a) a generic
point, (b) a horizon point, or (c) a horizon bifurcation point of a gST,
as classified in Proposition~\ref{prp:gST-isom}.
\end{prp}

\begin{proof}
We address the last statement first. The transitions between the
different charts in (a), (b) and (c) are effected with the help of the
\emph{tortoise coordinate}
\begin{equation} \label{eq:r-tortoise}
	r_* = \frac{1}{f'(r_H)} \log\left(\frac{r}{r_H}-1\right)
		+ \frac{h(r)}{f'(r_H)} ,
	\quad \text{which satisfies} \quad
	dr_* = \frac{dr}{f(r)} ,
\end{equation}
implicitly defining $h(r)$ as in~\eqref{eq:birkhoff-c}. The null
coordinate from (b) has the form $v = t+r_*$, while the double null
coordinates from (c) have the form $U = -e^{-f'(r_H)(t-r_*)/2}$, $V =
e^{f'(r_H)(t+r_*)/2}$. Direct calculation shows that the metrics
$g_{ab}$ expressed in these charts agree on overlaps. Hence, charts (b)
and (c) constitute analytic extensions of the charts from (a), which
when glued in a simply connected way, form the maximal analytic
extension, of the gST geometry from Definition~\ref{def:gST}. The
correspondence between points (a), (b) and (c) from the current
Proposition with those from Proposition~\ref{prp:gST-isom} is obvious.

Before entering further specific arguments, we use
formulas~\eqref{eq:wp-ricc} and~\eqref{eq:ddr} to project the Einstein
equations~\eqref{eq:ee} onto the base factor of the warped product:
\begin{equation}
	R_{ab} - (n-2)\frac{\nabla_a\nabla_b r}{r} -
	\frac{2\Lambda}{(n-2)} g_{ab} = 0 .
\end{equation}
Recalling that in $2$ dimensions $R_{ab} = \frac{1}{2} R g_{ab}$, the
equation decomposes into its trace and traceless parts:
\begin{equation}
	R - (n-2)\frac{\square r}{r} - \frac{4\Lambda}{(n-2)} = 0 , \quad
	(n-2)\frac{\nabla_a \nabla_b r}{r} - \frac{\square r}{2r} g_{ab} = 0 .
\end{equation}
Contracting the traceless part with $\eps_{ab}$ shows that $t_a =
-\eps_{ab} r^b$ is Killing, $\nabla_{(a} t_{b)} = 0$.

Suppose that $x$ is critical a point of $r$, that is, $\nabla_a r (x)= 0$,
and hence also $t_a(x) = 0$. We will now argue that this critical point
must be non-degenerate and hence isolated (cf.\ Remark~2.9
in~\cite{an-wong}). From the projected Einstein equations above,
$\nabla_a \nabla_b r \propto g_{ab}$, hence either $\nabla_a \nabla_b r(x)
= 0$ or the critical point is non-degenerate. If indeed $\nabla_a
\nabla_b r(x) = 0$, then the formula $t_a = - \eps_{ab} \nabla^b r$
tells us that $t_a(x) = 0$ and $\nabla_a t_b(x) = 0$ as well. In turn,
this implies that locally $t_a \equiv 0$, and hence also $\nabla_a r
\equiv 0$, which contradicts our hypothesis that $r$ is not locally
constant. The reason is that $t_a$ solves the Killing equation, which is
an equation of \emph{finite type}. In short, knowing $t_a(x)$ and
$\nabla_a t_b(x)$ determines $t_a$ uniquely in a neighborhood of
$x$~\cite[App.B]{geroch-killing}, which in this case gives $t_a \equiv
0$.

Next, we address each of the possibilities stated in the theorem. The
function $f(r)$ from~\eqref{eq:f-birkhoff} always appears as the general
solution, parametrized by constants $\alpha$ and $M$, of the
differential equation
\begin{equation}
	r\left(rf' + (n - 3)f\right)' = - \frac{4\Lambda}{(n-2)} r^2 .
\end{equation}

(a)
When $\nabla_a r \ne 0$ is non-null, we are free to pick orthogonal
coordinates $(t,r)$, with $(\del_t)^a \propto t^a$ Killing. Then, the most
general ansatz for the metric is $g_{ab} = -f(r) dt_a dt_b + 1/h(r) dr_a
dr_b$. The projected Einstein equations reduce to
\begin{equation}
	\frac{h'}{h} = \frac{f'}{f} , \quad
	r\left(rf' + (n - 3)f\right)' = - \frac{4\Lambda}{(n-2)} r^2 .
\end{equation}
Up to rescaling $t$ by a constant, $h(r) = f(r)$ and $f(r)$ is as
stipulated. The metric $g_{ab}$ is singular only when $\alpha = M =
\Lambda = 0$, meaning that all other values of the parameters are
allowed.

(b)
When $\nabla r(x) \ne 0$ is null, we are free to pick coordinates
$(v,r)$, with $\nabla_a v$ null everywhere, with $(\del_v)^a \propto t^a$
Killing. Then, the most general ansatz for the metric is $g_{ab} = -f(r)
dv_a dv_b + 2 h(r) dv_{(a} dr_{b)}$. The projected Einstein equations
reduce to
\begin{equation}
	h' = 0, \quad
	r(rf' + (n-3)f)' = -\frac{4\Lambda}{(n-2)} r^2 .
\end{equation}
Up to rescaling $v$ by a constant, $h(r) = 1$ and $f(r)$ is as
stipulated. For $\nabla_a r$ to be null at $x$ as well, we must have
$f(r_H) = 0$.

(c)
When $x$ is a non-degenerate critical point of $r$, we are free to pick
double null coordinates $(U,V)$ centered at $x$. Let $r_H = r(x)$, which
must be a non-zero constant. In $2$ dimensions, double null coordinates
are unique up to permutation and individual reparametrization of each
coordinate. Then, the most general ansatz for the metric is
\begin{equation}
	g_{ab} = 2 F(U,V) \, dU_{(a} dV_{b)} .
\end{equation}
Our hypotheses on $r$ force its Taylor expansion to start, up to a
constant rescaling, as
\begin{equation}
	r(U,V)
	= r_H + \frac{1}{2L} g_{ab} (U\del_U)^a (V\del_V)^b + O^3(U,V)
	= r_H + \frac{F(0,0)}{L} UV + O^3(U,V) ,
\end{equation}
for some constant $L\ne 0$, which will be constrained later on. The
precise form of $r=r(U,V)$ is to be determined from the equations. The
traceless part of the projected Einstein equations reduces to
\begin{equation} \label{eq:birkhoff-trless}
	\del_U\frac{\del_U r}{F} = 0 , \quad
	\del_V\frac{\del_V r}{F} = 0
	\iff
	\frac{\del_U r}{\xi(U)} = V \frac{F(U,V)}{L\xi(U) \eta(V)}, \quad
	\frac{\del_V r}{\eta(V)} = U \frac{F(U,V)}{L\xi(U) \eta(V)} ,
\end{equation}
for some arbitrary $\xi(U)$ and $\eta(V)$, though with $\xi(0) = 1 =
\eta(0)$ as needed to maintain our hypotheses on the Taylor expansion of
$r$. We are free to change our ansatz by $F(U,V) \mapsto F(U,V)
\xi(U)\eta(V)$ and reparametrize the coordinates subject to $\xi(U) dU
\mapsto dU$, $\eta(V) dV \mapsto dV$, effectively setting $\xi = 1 =
\eta$. Then, two immediate integrability conditions follow:
\begin{gather}
	U (\del_U r - V F) - V (\del_V r - U F)
	= U\del_U r - V\del_V r = 0 ,
	\\
	U\del_U(\del_V r - U F) - V\del_V(\del_U r - V F)
	= UV (U\del_U F - V\del_V F) = 0 .
\end{gather}
Therefore, both $F$ and $r$ are constant along the flow lines of the
vector field $U\del_U - V\del_V$, which are the connected components of
the level sets of $UV$. Without loss of generality, we can restrict to a
neighborhood where each level set of $UV$ consists of exactly two
connected components, exchanged by the flip $(U,V) \mapsto (-U,-V)$.

At this point, we would like to conclude that $r=r(UV)$ and $F=F(UV)$
for some smooth functions $r(z)$ and $F(z)$, but this conclusion must be
postponed due to the technical complication (not shared by
\emph{polynomial} or \emph{analytic functions}) that a \emph{smooth
function} invariant under the flow of $U\del_U - V\del_V$ takes such a
form only on those open regions where the product $UV$ may play the role
of a simple coordinate (no critical points, connected level sets). For
instance, $F = F_{\{U>0\}}(UV)$ on $U>0$ and $F = F_{\{V>0\}}(UV)$ on
$V>0$, but $F_{\{U>0\}}(z)$ and $F_{\{V>0\}}(z)$ may be different smooth
functions. Of course, these functions have to agree on overlapping
regions, namely for $z>0$, in this case. Below, we will presume that we
are restricting to one of the open regions $U>0$, $U<0$, $V>0$, or
$V<0$.

It turns out that it is more convenient to write everything in terms of
$r$, $UV = UV(r)$ and $F=F(r)$, which is always possible to do locally,
away from $(U,V)=(0,0)$ and subject to the above caveats. Taking
advantage of the usual identity $(UV)' = 1/r'$, our previous integrating
step~\eqref{eq:birkhoff-trless} simply gives $F = L/(UV)'$. Thus, the
remaining trace part of the projected Einstein equations reduces to
\begin{equation}
	r\del_r (r\del_r + (n-3)) \frac{2}{L}\frac{UV}{(UV)'}
		= -\frac{4\Lambda}{(n-2)} r^2
	\iff
	\frac{UV}{(UV)'} = \frac{L}{2} f(r) ,
\end{equation}
with $f(r)$ as stipulated. At $(U,V)=0$, the above left-hand side
evaluates to $0$. Hence, $r=r_H$ must solve $f(r) = 0$. If it is a
multiple root, that is $f(r) = C(r-r_H)^k + O(r-r_H)^{k+1}$ for some
constants $C$ and $k>1$, then asymptotically $UV \sim
e^{D/(r-r_H)^{k-1}}$ for some constant $D$ as $r\to r_H$, which is not
compatible with our requirement that $UV$ be a smooth function of $r$ at
$r=r_H$ and vice versa. Hence, $r=r_H$ must also be a \emph{simple root}
of $f(r)=0$, meaning that $f'(r_H) \ne 0$.

Thus, given that $f(r) = 0$ has a simple root at $r_H$, we can rewrite
the equation for $UV$ as
\begin{equation}
	\frac{(UV)'}{UV} = \frac{2}{L f'(r_H)(r-r_H)} + h'(r)
	\iff
	UV = C(r/r_H-1)^{\frac{2}{Lf'(r_H)}} e^{h(r)} ,
\end{equation}
for some constant $C$, with smooth
\begin{equation}
	h(r) = \frac{2}{L f'(r_H)} \int_{r_H}^r ds \left(\frac{f'(r_H)}{f(s)} - \frac{1}{s-r_H}\right) .
\end{equation}
For this formula to be consistent with the Taylor
expansion $r(U,V) = r_H + \frac{F(0,0)}{L} UV + O^3(U,V)$, we must have
$L = 2/f'(r_H)$ and $C = r_H L/F(0,0)$. The final form of the solution
is then
\begin{equation} \label{eq:birkhoff-UV-F}
	UV = \frac{2 r_H}{f'(r_H) F(0,0)}
		\left(\frac{r}{r_H} - 1\right) e^{h} , \quad
	F = \frac{L}{(UV)'}
		= \frac{f'(r_H) F(0,0)}{2r_H} \frac{-2f e^{-h}}{f'(r_H)^2 (1-r/r_H)} .
\end{equation}
To bring the metric into the desired form, it remains only to choose the
value of $F(0,0) = \frac{2r_H}{f'(r_H)}$, which could be done by
appropriately rescaling $U$ or $V$ by a constant. This also finally
fixes the initial coefficients in Taylor expansion $r = r_H +
r_H UV + O^3(U,V)$.

Finally, recall that the above discussion, determining the precise form
of $UV=UV(r)$ and $F=F(r)$, applies separately in each of the open
regions $U>0$, $U<0$, $V>0$ or $V<0$, though that precise form of the
functions $UV(r)$ and $F(r)$ may differ from region to region. It is
obvious that the only differences may be in the values of the constants
$\alpha$ and $M$, which a priori may be different in these different
regions. However, $UV(r)$ and $F(r)$ must agree on the intersection
whenever two of these regions overlap (e.g.,\ $U>0$ and $V>0$), and this
is only possible if the values of $\alpha$ and $M$ agree between the
overlapping regions. Considering all possible overlaps, the values of
$\alpha$ and $M$ must then agree in all these regions. In other words,
the formulas in~\eqref{eq:birkhoff-UV-F} actually hold on a whole
neighborhood of $(U,V)=(0,0)$ without any more reservations (the
extension to the origin is by continuity). This concludes the proof.
\end{proof}

\section{Characterization} \label{sec:ideal}

In this Section, we state and prove our main result on the IDEAL
characterization of local isometry classes
(Definition~\ref{def:loc-isom}) of generalized Schwarzschild-Tangherlini
(gST) geometries (Definition~\ref{def:gST}). The result comes in two
versions, one valid for any dimension $n\ge 5$
(Theorem~\ref{thm:ideal-gST}), and the other valid for $n=4$ as well as
higher dimensions (Theorem~\ref{thm:ideal-gST-4d}). The only difference
between them is the covariant formula for extracting the idempotent
projector $\bar{\Omega}_{\mu\nu}$ from the curvature. In higher
dimensions, $n\ge 5$, it can be obtained by a simpler formula than in
$n=4$. However, the more complicated formula that works in $n=4$ also
generalizes to higher dimensions. When restricted to the standard
spherically symmetric, $\Lambda = 0$, $n=4$ Schwarzschild solution, our
Theorem~\ref{thm:ideal-gST-4d} provides an independent alternative IDEAL
characterization compared to the one previously obtained
in~\cite{fs-schw}.%
	\footnote{The notation in~\cite{fs-schw} is somewhat hard to follow.
	A transcription of the key formulas into more standard tensor
	notation can be found in~\cite{gpgl-schw}.} %
For other values of the dimension $n$, and the parameters $\alpha$, $M$
and $\Lambda$, the results of this Section are new.

\begin{thm} \label{thm:ideal-gST}
Consider a Lorentzian geometry $(\bar{\mathcal{M}}, \bar{g})$,
with $\dim \bar{\mathcal{M}} = n \ge 5$. Given a constant $\Lambda$,
define the following tensors and scalars from the metric and the
curvature:
\begin{subequations}
\begin{align}
	\bar{T}_{\mu\nu\la\ka}
	&:= \bar{R}_{\mu\nu\la\ka}
		- \frac{\Lambda}{(n-1)(n-2)} (\bar{g}\odot \bar{g})_{\mu\nu\la\ka} , \\
	\rho
	&:= \left[\frac{(\bar{T}\cdot \bar{T} \cdot \bar{T})_{\mu\nu}{}^{\mu\nu}}
			{8(n-1)(n-2)(n-3) [(n-2)(n-3)(n-4)+2]}\right]^{\frac{1}{3}} , \\
	\ell_\mu
	&:= -\frac{1}{(n-1)} \frac{\bnabla_\mu \rho}{\rho} , \\
	A
	&:= \ell_\mu \ell^\mu + 2\rho + \frac{2\Lambda}{(n-1)(n-2)} , \\
\label{eq:Omega-gen}
	\bar{\Omega}_{\mu\nu}
	&:= \frac{2(n-2)^2}{(n-1)(n-4)} 
		\frac{(\bar{T}\cdot \bar{T})_{\mu\la}{}^{\la}{}_\nu}
			{(\bar{T}\cdot\bar{T})_{\la\ka}{}^{\la\ka}}
		+ \frac{(n-2)(n-3)}{(n-1)(n-4)} \bar{g}_{\mu\nu} , \\
	g_{\mu\nu}
	&:= \bar{g}_{\mu\nu} - \bar{\Omega}_{\mu\nu} , \\
	Z_{\mu\nu\la\ka}
\notag
	&:= \bar{T}_{\mu\nu\la\ka} \\
	& \qquad {}
		- \rho \left[
			\frac{(n-2)(n-3)}{2} (g\odot g)_{\mu\nu\la\ka}
			+ (\bar{\Omega} \odot \bar{\Omega})_{\mu\nu\la\ka}
			- (n-3) (g\odot \bar{\Omega})_{\mu\nu\la\ka} \right] .
\end{align}
\end{subequations}
Then the geometry $(\bar{\mathcal{M}}, \bar{g})$ is locally isometric at
$\bar{x}\in \bar{\M}$ to a gST geometry
$(\bar{\M}',\bar{g}')_{\alpha',M',\Lambda'}$ (Definition~\ref{def:gST})
iff $\Lambda=\Lambda'$ and the following conditions hold on some
neighborhood of $\bar{x}$:
\begin{subequations} \label{eq:ideal-gST}
\begin{gather}
\label{eq:ideal-gST-ineq}
	\rho \ne 0 , 
		\quad
	\ell_\mu \ne 0 ,
		\quad
	(\bar{T}\cdot\bar{T})_{\la\ka}{}^{\la\ka} \ne 0 ,
		\quad
	\bar{\Omega}_{\mu\nu} \ge 0 ,
		\\
\label{eq:ideal-gST-wp}
\begin{gathered}
	\bnabla_{[\mu} \ell_{\nu]} = 0 ,
		\quad
	\ell^\mu \bar{\Omega}_{\mu\nu} = 0 ,
		\\
	\bar{\Omega}_\mu{}^\mu = (n-2) ,
		\quad
	\bar{\Omega}_\mu{}^\nu \bar{\Omega}_{\nu\la} = \bar{\Omega}_{\mu\la} ,
		\quad
	\bnabla_\mu \bar{\Omega}_{\nu\la} = - 2\bar{\Omega}_{\mu(\nu} \ell_{\la)} ,
\end{gathered}
		\\
\label{eq:ideal-gST-mass}
	Z_{\mu\nu\la\ka} = 0 ,
		\quad
	(\sgn A) \left|A\right|^{n-1} \rho^{-2}
		= (\sgn \alpha') \left|\alpha'\right|^{n-1} M'^{-2} .
\end{gather}
\end{subequations}
\end{thm}

By the inequalities $\rho \ne 0$, $\ell_\mu \ne 0$ and
$(\bar{T}\cdot\bar{T})_{\la\ka}{}^{\la\ka} \ne 0$, we mean that these
objects do not vanish at $\bar{x}\in \M$ and hence, by continuity, in
some neighborhood of $\bar{x}$.
By the inequality $\bar{\Omega}_{\mu\nu} \ge 0$ we mean that the
quadratic form $\bar{\Omega}(v,v) = \bar{\Omega}_{\mu\nu} v^\mu v^\nu$
is positive-semidefinite.

Note that, in choosing the precise form of the
conditions~\eqref{eq:ideal-gST}, we have aimed for clarity rather than
any particular kind of minimality. So, for instance, the condition
$\bar{\Omega}_{\mu}{}^{\mu} = (n-2)$ is automatically satisfied by
virtue of our definition of $\bar{\Omega}_{\mu\nu}$, the same being true
for the condition $\bnabla_{[\mu} \ell_{\nu]} = 0$.

\begin{proof}
In the easy direction, the direct calculations from
Section~\ref{sec:gST} show that all of the
identities~\eqref{eq:ideal-gST} hold for any gST geometry.

In the other direction, we first note that the
conditions~\eqref{eq:ideal-gST-wp} involving $\ell_\mu$ and
$\bar{\Omega}_{\mu\nu}$ are precisely needed by
Proposition~\ref{prp:wp-char} to locally put the geometry into
$2+m$-warped product form $(\bar{\M},\bar{g}) = (\M,g) \times_r
(\S,\Omega)$, with $m=n-2$ and $\bar{g}_{\mu\nu} = g_{\mu\nu} + r^2
\Omega_{\mu\nu}$, where $\Omega_{\mu\nu} = r^{-2} \bar{\Omega}_{\mu\nu}$
and $\ell_\mu = \bnabla_\mu \log |r|$, $r$ not locally constant by
$\ell_\mu \not\equiv 0$. Since, $\bar{\Omega}_{\nu\nu}\ge 0$ and
$\bar{g}_{\mu\nu}$ is Lorentzian, $g_{\mu\nu}$ must be a Lorentzian
metric when restricted to the base factor of the warped product.

Next, taking the trace of the $Z_{\mu\nu\la\ka} = 0$ identity, we obtain
precisely the $\Lambda$-vacuum Einstein equations, forcing the equality
$\Lambda = \Lambda'$. Appealing to Birkhoff's theorem
(Proposition~\ref{prp:birkhoff}), we can conclude that the metric
$g_{ab}$ on the \emph{base} of the warped product has the gST form
$(\M,g)_{n,\alpha,M,\Lambda}$ (Section~\ref{sec:gST}), for some values
of the parameters $\alpha$ and $M$ (we have not yet drawn any conclusion
about the \emph{warped factor}). Note that our version of Birkhoff's
theorem is applicable as long as the warping function $r$ is not locally
constant at $\bar{x}\in \bar{\M}$, without other restrictions on
$\bnabla_\mu r(\bar{x})$, with the different possibilities listed as
parts (a), (b) and (c) in Proposition~\ref{prp:birkhoff}.

Two immediate consequences, again following the direct calculations
in Section~\ref{sec:gST}, are the identities
\begin{equation*}
	\rho = \frac{M}{r^{n-1}}
	\quad \text{and} \quad
	A = \frac{\alpha}{r^2} .
\end{equation*}

Now, knowing that our geometry is locally of $2+m$-warped product form,
implies that its Riemann tensor takes the form~\eqref{eq:wp-riem}, where
we can replace $r_\sigma r^\sigma / r^2 = \ell^2$. Hence, projecting the
identity $Z_{\mu\nu\la\ka} = 0$ by $\bar{\Omega}_{\mu\nu}$ on each
index, we obtain
\begin{equation}
	\bar{\Omega}_{\mu'}^\mu
	\bar{\Omega}_{\nu'}^\nu
	\bar{\Omega}_{\la'}^\la
	\bar{\Omega}_{\ka'}^\ka
	\left( r^2 S_{\mu\nu\la\ka}
		- \frac{A}{2}
			(\bar{\Omega} \odot \bar{\Omega})_{\mu\nu\la\ka} \right)
	= 0 .
\end{equation}
Substituting in what we already know about $\bar{\Omega}_{\mu\nu}$ and
$A$ into this formula, it reduces to the equality $S_{ABCD} =
\frac{\alpha}{2} (\Omega \odot \Omega)_{ABCD}$ on the warped factor
$(\mathcal{S},\Omega)$. In other words, the \emph{warped} factor is
locally of constant curvature, with sectional curvature $\alpha$. Hence,
our geometry $(\bar{\M}, \bar{g})$ is indeed locally isometric at
$\bar{x} \in \bar{M}$ to a gST geometry $(\bar{\M}',
\bar{g}')_{\alpha,M,\Lambda}$ with parameters $\alpha$, $M$ and
$\Lambda$ (Definition~\ref{def:gST}).

Finally, referring again to the direct calculations from
Section~\ref{sec:gST}, and in particular the
identity~\eqref{eq:gST-Arho}, the last identity
from~\eqref{eq:ideal-gST} implies the equality
\begin{equation}
	(\sgn \alpha) \left|\alpha\right|^{n-1} M^{-2}
	=
	(\sgn \alpha') \left|\alpha'\right|^{n-1} M'^{-2} .
\end{equation}
So, invoking Proposition~\ref{prp:gST-isom}, we can at last conclude
that the gST geometry that we have identified locally about $\bar{x} \in
\bar{\M}$ is indeed isometric to the desired reference gST geometry,
$(\bar{\M}', \bar{g}')_{\alpha,M,\Lambda} \cong (\bar{\M}',
\bar{g}')_{\alpha',M',\Lambda'}$.
\end{proof}

For a version of the above result that holds also when $n=4$ we need
only replace formula~\eqref{eq:Omega-gen} for $\bar{\Omega}_{\mu\nu}$ by
formula~\eqref{eq:Omega-4d} below (recall the discussion around
Equations~\eqref{eq:gST-Omega} and~\eqref{eq:gST-Omega-4d} in
Section~\ref{sec:gST}). Hence, the proof of the following result
proceeds in an exactly analogous way.

\begin{thm} \label{thm:ideal-gST-4d}
Consider a Lorentzian geometry $(\bar{\mathcal{M}}, \bar{g})$, with
$\dim \bar{\mathcal{M}} = n \ge 4$. Then the same statement as in
Theorem~\ref{thm:ideal-gST} holds, with the exception that we must
change the definition
\begin{equation} \label{eq:Omega-4d}
	\bar{\Omega}_{\mu\nu}
	:= -\frac{1}{(n-1)(n-3) \, \rho \, \ell^2} \left(\bar{T}_{\mu\la\nu\ka}
		- \frac{(n-2)(n-3)}{2} \rho \, (\bar{g}\odot \bar{g})_{\mu\la\nu\ka} \right)
		\ell^\la \ell^\ka ,
\end{equation}
while adding the hypotheses that $\ell^2 \not\equiv 0$ on a neighborhood
of $\bar{x}\in \bar{\M}$ and that $\bar{\Omega}_{\mu\nu}$ extends by
continuity to a smooth tensor field on this neighborhood despite
$\ell^2$ possibly vanishing at some points.
\end{thm}

\section{Discussion} \label{sec:discuss}

We have given an IDEAL characterization (Theorems~\ref{thm:ideal-gST},
\ref{thm:ideal-gST-4d}) of each spacetime from the family of local
isometry classes \emph{generalized Schwarzschild-Tangherlini (gST)}
spacetimes (Definition~\ref{def:gST}), which consists of maximally
symmetric $\Lambda$-vacuum $2+m$-warped products. In particular, this
family includes the higher dimensional spherically symmetric black
holes, which generalize the $4$-dimensional Schwarzschild solution and
which were first investigated by Tangherlini~\cite{tangherlini}.

Our strategy, inspired by the related recent work on the
characterization of cosmological FLRW spacetimes~\cite{cdk}, was to
first identify a geometric characterization of the $2+m$-warped product
structure in terms of a rank-$m$ symmetric projector $\bar{\Omega}$
(Proposition~\ref{prp:wp-char})~\cite{fs-2+2, gpgl-biconf} and then to identify a
covariant formula for $\bar{\Omega}$ in terms of the curvature of a
given gST geometry. The previously existing IDEAL characterization of
the $4$-dimensional Schwarzschild geometry~\cite{fs-schw, gpgl-schw} relied much
more on an intricate algebraic classification of the Riemann tensor,
special to $4$ dimensions. Unfortunately, we could not generalize the
latter approach to higher dimensions directly. On the other hand, our
general strategy succeeds also in $4$ dimensions
(Theorem~\ref{thm:ideal-gST-4d}) and thus provides an alternative
characterization of the Schwarzschild geometry, which should be compared
to that of~\cite{fs-schw}. We leave such a comparison to future work.

As discussed in~\cite{fhk}, the linearization of the tensors of an IDEAL
characterization of a given reference geometry provides a set of
gauge-invariants with respect to linearized gauge transformations
(diffeomorphisms) of linearized gravity on that geometry. Heuristically,
this set of invariants is a good candidate for being complete, but to be
rigorous its completeness should be proven separately. In the recent
work~\cite{kh-compat}, we have explicitly exhibited (by a different
method) complete sets of linear invariants for each geometry in the gST
family. Relating these invariants to the linearization of the IDEAL
characterization tensors, as well as vice versa, can accomplish two
goals: give a geometric interpretation to the invariants
of~\cite{kh-compat} and to prove the completeness of the linearized
invariants that can be obtained from the present work.

A natural direction for related future work is to extend it to an IDEAL
characterization of other black hole spacetimes. For instance, the
generalization to charged spherical symmetric black holes, the
Reissner-Nordstr\"om geometry and its higher dimensional generalizations,
should be straightforward. A bigger challenge would be to generalize it
to higher dimensional rotating black holes, the Myers-Perry
generalizations of the Kerr geometry, perhaps building on the existing
characterization of the $4$-dimensional Kerr spacetime~\cite{fs-kerr}.
Eventually, it would be interesting to extend the characterization to
the full Kerr-Taub-NUT-(A)dS family~\cite{fkk-review} and higher
dimensional versions.

An important future application of the above results could be an
intrinsic and invariant characterization of asymptotic flatness.
Usually, asymptotic flatness is defined by an asymptotic condition on
the metric in a special coordinate system. On the other hand, this
definition is supposed to capture the asymptotic approach to flatness or
the asymptotic end of a black hole spacetime. Thus, having an on hand an
IDEAL characterization of these reference geometries may give us a
chance to intrinsically and invariantly define asymptotic approach to
them, providing an alternative definition of asymptotic flatness.

\paragraph{Acknowledgments.}
Thanks to Willie Wong and Markus Fr\"ob for helpful discussions, as well
as to Alfonso Garc\'{\i}a-Parrado G\'{o}mez-Lobo for bibliographic
suggestions. Thanks also to the anonymous referees for suggesting
improvements. The author was partially supported by the GA\v{C}R project
18-07776S and RVO: 67985840.

\bibliographystyle{utphys-alpha}
\bibliography{killing}

\providecommand{\href}[2]{#2}\begingroup\raggedright\begin{thebibliography}{10}

\bibitem{an-wong}
X.~An and W.~W.~Y. Wong, ``Warped product space-times,''
  \href{http://dx.doi.org/10.1088/1361-6382/aa8af7}{{\em Classical and Quantum
  Gravity} {\bfseries 35} (2018) 025011},
  \href{http://arxiv.org/abs/1707.01483}{{\ttfamily arXiv:1707.01483}}.

\bibitem{cdk}
G.~Canepa, C.~Dappiaggi, and I.~Khavkine, ``{IDEAL} characterization of
  isometry classes of {FLRW} and inflationary spacetimes,''
  \href{http://dx.doi.org/10.1088/1361-6382/aa9f61}{{\em Classical and Quantum
  Gravity} {\bfseries 35} (2018) 035013},
  \href{http://arxiv.org/abs/1704.05542}{{\ttfamily arXiv:1704.05542}}.

\bibitem{coll-ferrando}
B.~Coll and J.~J. Ferrando, ``Thermodynamic perfect fluid. {Its} {Rainich}
  theory,'' \href{http://dx.doi.org/10.1063/1.528477}{{\em Journal of
  Mathematical Physics} {\bfseries 30} (1989) 2918--2922}.

\bibitem{dotti-schw}
G.~Dotti, ``Nonmodal linear stability of the {Schwarzschild} black hole,''
  \href{http://dx.doi.org/10.1103/physrevlett.112.191101}{{\em Physical Review
  Letters} {\bfseries 112} (2014) 191101},
  \href{http://arxiv.org/abs/1307.3340}{{\ttfamily arXiv:1307.3340}}.

\bibitem{ehlers-kundt}
J.~Ehlers and W.~Kundt, ``Exact solutions of the gravitational field
  equations,'' in {\em Gravitation: an introduction to current research},
  L.~Witten, ed., pp.~49--101.
\newblock Wiley, New York, 1962.

\bibitem{fs-typeD}
J.~J. Ferrando and J.~A. S\'{a}ez, ``On the classification of {type D}
  spacetimes,'' \href{http://dx.doi.org/10.1063/1.1640795}{{\em Journal of
  Mathematical Physics} {\bfseries 45} (2002) 652--667},
  \href{http://arxiv.org/abs/gr-qc/0212086}{{\ttfamily arXiv:gr-qc/0212086}}.

\bibitem{fs-schw}
J.~J. Ferrando and J.~A. S\'{a}ez, ``An intrinsic characterization of the
  {Schwarzschild} metric,''
  \href{http://dx.doi.org/10.1088/0264-9381/15/5/014}{{\em Classical and
  Quantum Gravity} {\bfseries 15} (1998) 1323--1330}.

\bibitem{fs-kerr}
J.~J. Ferrando and J.~A. S\'{a}ez, ``An intrinsic characterization of the
  {Kerr} metric,'' \href{http://dx.doi.org/10.1088/0264-9381/26/7/075013}{{\em
  Classical and Quantum Gravity} {\bfseries 26} (2009) 075013},
  \href{http://arxiv.org/abs/0812.3310}{{\ttfamily arXiv:0812.3310}}.

\bibitem{fs-2+2}
J.~J. Ferrando and J.~A. S\'{a}ez, ``An intrinsic characterization of 2+2
  warped spacetimes,''
  \href{http://dx.doi.org/10.1088/0264-9381/27/20/205023}{{\em Classical and
  Quantum Gravity} {\bfseries 27} (2010) 205023},
  \href{http://arxiv.org/abs/1005.1491}{{\ttfamily arXiv:1005.1491}}.

\bibitem{fs-sphsym}
J.~J. Ferrando and J.~A. S\'{a}ez, ``An intrinsic characterization of
  spherically symmetric spacetimes,''
  \href{http://dx.doi.org/10.1088/0264-9381/27/20/205024}{{\em Classical and
  Quantum Gravity} {\bfseries 27} no.~20, (2010) 205024},
  \href{http://arxiv.org/abs/1005.1780}{{\ttfamily arXiv:1005.1780}}.

\bibitem{fs-ricci}
J.~J. Ferrando and J.~A. S\'{a}ez, ``Labeling spherically symmetric spacetimes
  with the {Ricci} tensor,''
  \href{http://dx.doi.org/10.1088/1361-6382/aa525d}{{\em Classical and Quantum
  Gravity} {\bfseries 34} (2017) 045002},
  \href{http://arxiv.org/abs/1701.05023}{{\ttfamily arXiv:1701.05023}}.

\bibitem{fhk}
M.~B. Fr\"{o}b, T.-P. Hack, and I.~Khavkine, ``Approaches to linear local
  gauge-invariant observables in inflationary cosmologies,''
  \href{http://dx.doi.org/10.1088/1361-6382/aabcb7}{{\em Classical and Quantum
  Gravity} {\bfseries 35} (2018) 115002},
  \href{http://arxiv.org/abs/1801.02632}{{\ttfamily arXiv:1801.02632}}.

\bibitem{fkk-review}
V.~P. Frolov, P.~Krtou\v{s}, and D.~Kubiz\v{n}\'{a}k, ``Black holes, hidden
  symmetries, and complete integrability,''
  \href{http://dx.doi.org/10.1007/s41114-017-0009-9}{{\em Living Reviews in
  Relativity} {\bfseries 20} (2017) 6},
  \href{http://arxiv.org/abs/1705.05482}{{\ttfamily arXiv:1705.05482}}.

\bibitem{gpgl-biconf}
A.~Garc\'{\i}a-Parrado G\'{o}mez-Lobo, ``Bi-conformal vector fields and the
  local geometric characterization of conformally separable {pseudo-Riemannian}
  manifolds {I},'' \href{http://dx.doi.org/10.1016/j.geomphys.2005.06.005}{{\em
  Journal of Geometry and Physics} {\bfseries 56} (2006) 1069--1095},
  \href{http://arxiv.org/abs/math/0504162}{{\ttfamily arXiv:math/0504162}}.

\bibitem{gpgl-kerr}
A.~Garc\'{\i}a-Parrado G\'{o}mez-Lobo, ``Local non-negative initial data scalar
  characterization of the {Kerr} solution,''
  \href{http://dx.doi.org/10.1103/PhysRevD.92.124053}{{\em Physical Review D}
  {\bfseries 92} (2015) 124053},
  \href{http://arxiv.org/abs/1510.07561}{{\ttfamily arXiv:1510.07561}}.

\bibitem{gpgl-schw}
A.~Garc\'{\i}a-Parrado G\'{o}mez-Lobo and J.~A. Valiente~Kroon, ``Initial data
  sets for the {Schwarzschild} spacetime,''
  \href{http://dx.doi.org/10.1103/physrevd.75.024027}{{\em Physical Review D}
  {\bfseries 75} (2007) 024027},
  \href{http://arxiv.org/abs/gr-qc/0609100}{{\ttfamily arXiv:gr-qc/0609100}}.

\bibitem{geroch-killing}
R.~Geroch, ``Limits of spacetimes,''
  \href{http://dx.doi.org/10.1007/bf01645486}{{\em Communications in
  Mathematical Physics} {\bfseries 13} (1969) 180--193}.

\bibitem{kh-calabi}
I.~Khavkine, ``The {Calabi} complex and {Killing} sheaf cohomology,''
  \href{http://dx.doi.org/10.1016/j.geomphys.2016.06.009}{{\em Journal of
  Geometry and Physics} {\bfseries 113} (2017) 131--169},
  \href{http://arxiv.org/abs/1409.7212}{{\ttfamily arXiv:1409.7212}}.

\bibitem{kh-compat}
I.~Khavkine, ``Compatibility complexes of overdetermined {PDEs} of finite type,
  with applications to the {Killing} equation,''
  \href{http://dx.doi.org/10.1088/1361-6382/ab329a}{{\em Classical and Quantum
  Gravity} {\bfseries 36} (2019) 185012},
  \href{http://arxiv.org/abs/1805.03751}{{\ttfamily arXiv:1805.03751}}.

\bibitem{ki-master}
H.~Kodama and A.~Ishibashi, ``A master equation for gravitational perturbations
  of maximally symmetric black holes in higher dimensions,''
  \href{http://dx.doi.org/10.1143/ptp.110.701}{{\em Progress of Theoretical
  Physics} {\bfseries 110} (2003) 701--722},
  \href{http://arxiv.org/abs/hep-th/0305147}{{\ttfamily arXiv:hep-th/0305147}}.

\bibitem{krongos-torre}
D.~S. Krongos and C.~G. Torre, ``Geometrization conditions for perfect fluids,
  scalar fields, and electromagnetic fields,''
  \href{http://dx.doi.org/10.1063/1.4926952}{{\em Journal of Mathematical
  Physics} {\bfseries 56} (2015) 072503},
  \href{http://arxiv.org/abs/1503.06311}{{\ttfamily arXiv:1503.06311}}.

\bibitem{lake}
K.~Lake, ``Maximally extended, explicit and regular coverings of the
  {Schwarzschild-de Sitter} vacua in arbitrary dimension,''
  \href{http://dx.doi.org/10.1088/0264-9381/23/20/010}{{\em Classical and
  Quantum Gravity} {\bfseries 23} (2006) 5883--5895},
  \href{http://arxiv.org/abs/gr-qc/0507031}{{\ttfamily arXiv:gr-qc/0507031}}.

\bibitem{lobo-mimoso}
F.~S.~N. Lobo and J.~P. Mimoso, ``Possibility of hyperbolic tunneling,''
  \href{http://dx.doi.org/10.1103/physrevd.82.044034}{{\em Physical Review D}
  {\bfseries 82} (2010) 044034},
  \href{http://arxiv.org/abs/0907.3811}{{\ttfamily arXiv:0907.3811}}.

\bibitem{mars-local}
M.~Mars, \href{http://dx.doi.org/10.1007/978-3-319-60039-0_18}{``On local
  characterization results in geometry and gravitation,''} in {\em From Riemann
  to Differential Geometry and Relativity}, L.~Ji, A.~Papadopoulos, and
  S.~Yamada, eds., ch.~18, pp.~541--570.
\newblock Springer, Berlin, 2017.

\bibitem{martel-poisson}
K.~Martel and E.~Poisson, ``Gravitational perturbations of the {Schwarzschild}
  spacetime: A practical covariant and gauge-invariant formalism,''
  \href{http://dx.doi.org/10.1103/physrevd.71.104003}{{\em Physical Review D}
  {\bfseries 71} no.~10, (2005) 104003},
  \href{http://arxiv.org/abs/gr-qc/0502028}{{\ttfamily arXiv:gr-qc/0502028}}.

\bibitem{mobpm-kerr}
C.~Merlin, A.~Ori, L.~Barack, A.~Pound, and M.~van~de Meent, ``Completion of
  metric reconstruction for a particle orbiting a {Kerr} black hole,''
  \href{http://dx.doi.org/10.1103/physrevd.94.104066}{{\em Physical Review D}
  {\bfseries 94} (2016) 104066},
  \href{http://arxiv.org/abs/1609.01227}{{\ttfamily arXiv:1609.01227}}.

\bibitem{oneill}
B.~O'Neill, {\em Semi-Riemannian Geometry With Applications to Relativity},
  vol.~103 of {\em Pure and Applied Mathematics}.
\newblock Elsevier Science, 1983.

\bibitem{rainich}
G.~Y. Rainich, ``Electrodynamics in the general relativity theory,''
  \href{http://dx.doi.org/10.1090/s0002-9947-1925-1501302-6}{{\em Transactions
  of the American Mathematical Society} {\bfseries 27} (1925) 106--136}.

\bibitem{schleich-witt}
K.~Schleich and D.~M. Witt, ``A simple proof of {Birkhoff's} theorem for
  cosmological constant,'' \href{http://dx.doi.org/10.1063/1.3503447}{{\em
  Journal of Mathematical Physics} {\bfseries 51} no.~11, (2010) 112502},
  \href{http://arxiv.org/abs/0908.4110}{{\ttfamily arXiv:0908.4110}}.

\bibitem{schmidt}
H.-J. Schmidt, ``A new proof of {Birkhoff's} theorem,'' {\em Gravitation and
  Cosmology} {\bfseries 3} (1997) 185--190,
  \href{http://arxiv.org/abs/gr-qc/9709071}{{\ttfamily arXiv:gr-qc/9709071}}.
  \url{https://elibrary.ru/item.asp?id=12947329}.

\bibitem{stephani-sols}
H.~Stephani, D.~Kramer, M.~MacCallum, C.~Hoenselaers, and E.~Herlt,
  \href{http://dx.doi.org/10.1017/CBO9780511535185}{{\em Exact Solutions of
  {Einstein's} Field Equations}}.
\newblock Cambridge University Press, Cambridge, 2003.

\bibitem{stewart-walker}
J.~M. Stewart and M.~Walker, ``Perturbations of {Space-Times} in general
  relativity,'' \href{http://dx.doi.org/10.1098/rspa.1974.0172}{{\em
  Proceedings of the Royal Society of London. A. Mathematical and Physical
  Sciences} {\bfseries 341} (1974) 49--74}.

\bibitem{takeno}
H.~Takeno, ``On the spherically symmetric {Space-Times} in general
  relativity,'' \href{http://dx.doi.org/10.1143/ptp/8.3.317}{{\em Progress of
  Theoretical Physics} {\bfseries 8} (1952) 317--326}.

\bibitem{tangherlini}
F.~R. Tangherlini, ``{Schwarzschild} field in {$n$} dimensions and the
  dimensionality of space problem,''
  \href{http://dx.doi.org/10.1007/bf02784569}{{\em Il Nuovo Cimento} {\bfseries
  27} (1963) 636--651}.

\bibitem{taub}
A.~H. Taub, ``Empty space-times admitting a three parameter group of motions,''
  \href{http://dx.doi.org/10.2307/1969567}{{\em The Annals of Mathematics}
  {\bfseries 53} (1951) 472--490}.

\bibitem{wald-gr}
R.~M. Wald, {\em General Relativity}.
\newblock University of Chicago Press, Chicago, 1984.

\bibitem{wolf-cc}
J.~A. Wolf, {\em Spaces of Constant Curvature}, vol.~372 of {\em AMS-Chelsea}.
\newblock American Mathematical Society, 6th~ed., 2011.

\end{thebibliography}\endgroup

\end{document}